\documentclass[12pt]{amsart}
\usepackage{amsfonts, amsmath, amssymb}
\usepackage{amsthm}
\usepackage{graphicx}
\usepackage{float}
\usepackage{enumitem}

\usepackage[utf8]{inputenc} 
\usepackage{enumitem}
\usepackage{lipsum}  

\usepackage{graphicx}
\usepackage{subcaption}
\usepackage[utf8]{inputenc}

\newtheorem {Theorem}{Theorem}
\numberwithin{Theorem}{section}

\newtheorem {Proposition}[Theorem]{Proposition}

\theoremstyle{definition}
\newtheorem{Definition}[Theorem]{Definition}

\theoremstyle{remark}
\newtheorem{Remark}[Theorem]{Remark}

\expandafter\chardef\csname pre amssym.def
at\endcsname=\the\catcode`\@ \catcode`\@=11
\def\undefine#1{\let#1\undefined}
\def\newsymbol#1#2#3#4#5{\let\next@\relax
 \ifnum#2=\@ne\let\next@\msafam@\else
 \ifnum#2=\tw@\let\next@\msbfam@\fi\fi
 \mathchardef#1="#3\next@#4#5}
\def\mathhexbox@#1#2#3{\relax
 \ifmmode\mathpalette{}{\m@th\mathchar"#1#2#3}%
 \else\leavevmode\hbox{$\m@th\mathchar"#1#2#3$}\fi}
\def\hexnumber@#1{\ifcase#1 0\or 1\or 2\or 3\or 4\or 5\or 6\or 7\or 8\or
 9\or A\or B\or C\or D\or E\or F\fi}

\font\teneufm=eufm10 \font\seveneufm=eufm7 \font\fiveeufm=eufm5
\newfam\eufmfam
\textfont\eufmfam=\teneufm \scriptfont\eufmfam=\seveneufm
\scriptscriptfont\eufmfam=\fiveeufm

\catcode`\@=\csname pre amssym.def at\endcsname

\def    \p      {\partial}

\begin{document}





\title[Critical Community Size and Viral Strain Selection]{Deterministic Critical Community Size for the SIR system and Viral Strain Selection}
\author[M. F. Santos]{Marc\'ilio Ferreira dos Santos$^{*}$}
\author[C. Castilho]{C\'esar Castilho$^{\dagger}$}

\email{castilho@dmat.ufpe.br}

\bigskip

\maketitle

\centerline {$^{*}$ Núcleo de Formação de Docentes} \centerline{
Universidade Federal de Pernambuco} \par  \centerline{Caruaru, PE
CEP 55014-900 Brazil}\par \centerline{\email{marcilio.santos@ufpe.br}}\par
\vspace{0.5cm}
\centerline {$^\dagger$ Departamento de Matem\'atica} \par \centerline{
Universidade Federal de Pernambuco} \par  \centerline{Recife, PE
CEP 50740-540 Brazil}  \vspace{0.2cm}

\begin{abstract}
  In this paper the concept of Critical Community Size (CCS) for the deterministic SIR  model is introduced and its consequences for the disease dynamics are stressed.  The disease can fade out after an outburst. Also the principle of competitive exclusion holds no longer true. This is exemplified for the dynamics of two competing virus strains. The virus with higher $R_0$ can be eradicated from the population.
\end{abstract}

\vspace{0.4cm}
\centerline{ {\bf Key Words:} 
Epidemics, Virulence, SIR Modes, Evolution Theory.} \vspace{0.4cm} 
\centerline{{\bf AMS :}
92D30, 34C60, 37N25}

\section{Introduction}

  The  Critical Community Size (CCS) of a infectious disease is defined  as the minimum size of a closed population within which the disease's pathogen can persist \cite{bartlett1957measles,bartlett1960critical}. When the size of the  population is smaller then the CCS , the low density of infected hosts causes the extinction of the pathogen after an epidemic outbreak. In this case the disease is said to fade out \cite{anderson1979population,anderson1992infectious}. Classical SIR deterministic models for direct contact viral diseases \cite{brauer2019mathematical} fail to capture the fade out phenomena: either the number of infected converges to an endemic equilibrium after successive outbreaks or it disappears without any outbreak; the fate of the disease depending on a bifurcation parameter, called the basic reproductive number $R_0$ \cite{macdonald1952analysis,dietz1993estimation}.  \par
  
    One of the reasons for the SIR breakdown to capture the disease fade out for small populations is the  use of real numbers to count individuals. While being a good approximation for  large populations, counting individuals using real numbers has dramatic consequences when the number of individuals within a particular compartmental class becomes smaller then one and therefore extinct: the SIR model  fails to capture the small population extinction. In this paper extinction is incorporated into the model. When one of the variables representing the compartmental classes becomes smaller then  one it is immediately set to zero. This  has two important consequences. First
    the concept of Critical Community Size appears naturally. Second the principle of competitive exclusion no longer holds. \par 
    In a realistic context, when the number of individuals is too small, the probability of disease eradication is high. This Allé
    The paper is organized as follows. In SECTION \ref{SIR}
 the CCS for the classical SIR model with constant population is defined. The definition follows directly: for a population of $N$ individuals the disease will be eradicated if the density of  infected individuals becomes smaller then $\frac{1}{N}$. This simple fact,  that can be also interpreted as a consequence of the Alle effect \cite{allee1932studies,courchamp2008allee}, allows for the determination of the minimum viable population \cite{traill2007minimum} for a disease. This minimum viable population does not dependent on the disease $R_0$ value, instead it depends non trivially on the parameters reflecting the scaling properties of the SIR system. Curves for the CCS in terms of the parameters of the SIR system are exhibited. \par
 In SECTION \ref{competition} we study the consequences of the fade out on the dynamics of two different  competing virus strains. It is numerically shown that the Principle of Competitive Exclusion (PCE) \cite{bremermann1989competitive} is no longer true: A strain with smaller $R_0$ can eliminate one with higher $R_0$. On section \ref{conclusions} we draw our conclusions
\section{SIR model and Critical Community Size}
\label{SIR}
Let $S(t)$, $I(t)$ and $R(t)$ denote the number of susceptibles, infected
and removed at time $t$ respectively and $N(t) = S(t) + I(t) + R(t)$ be the total number of individuals. The SIR model states that
$$
\begin{array}{l}
S^{\, \prime} = - \beta \, \frac{S \, Y}{N} + \mu \, N - \mu \, S   \, ,\\
\\
Y^{\, \prime} = \beta \, \frac{S \, Y}{N} - \mu \, Y - \gamma \, Y  \, ,\\
\\
R^{\, \prime} = \gamma \, Y  - \mu \, R \, ,
\end{array}
$$
where $\beta$ is the infection rate, $\gamma$ is the clearance rate and
$\mu$ is the  mortality rate (assumed equal to the birth rate).
Adding the equations it follows that $N(t)$ is a first integral.
Introducing the densities $s = \frac{S}{N}$, $y = \frac{Y}{N}$ and $r = \frac{R}{N}$ , equations become
\begin{equation}
\begin{array}{l}
s^{\, \prime} = - \beta \, s \, y  + \mu \, ( 1 - s )   ,\\
\\
y^{\, \prime} = \beta \, s \, y  - (\mu + \gamma) \, y \, .\\
\end{array}
\label{eq.sir_normalizado}
\end{equation}
The equation for $r(t)$ is omitted since $r(t) = 1 - s(t) - y(t) \, .$

  \par

The qualitative dynamics of the above system is determined by the bifurcation parameter
$$ R_0 = \frac{\beta}{\gamma + \mu } $$
called the disease basic reproductive number. It represents the average number of cases caused by one infected individual in a totally susceptible
population. The following are well known facts \cite{brauer2019mathematical}. \par \bigskip
i) If $ R_0 < 1 $ the dynamical system has only one equilibrium
point $E_0=(1,0)$ called the disease free equilibrium point. $E_0$ is a globally stable critical point. \par \bigskip 
ii) If $R_0 > 1$ the dynamical system has two equilibria. $E_0$ an unstable critical point and $E_1 =(s^*,y^*) = \left( \frac{1}{R_0} \, , \, \frac{\mu}{\beta + \mu} \, \left( 1 - 1/R_0 \right) \right)$  a globally stable critical point. The global stability of $E_1$ can be proved using the Lyapounov function \cite{korobeinikov2002lyapunov}
\begin{equation}
V(s,y) = s^* \, \left( \frac{s}{s^*} - \ln\left(\frac{s}{s^*}\right) \right)
+ y^* \left(\frac{y}{y^*} - \ln\left(\frac{y}{y^*}\right)  \right) \, .
\end{equation}
Since $N$  represents the total number of individuals one must have  $ N \in \mathbb{Z}^+$. Accordingly, the smallest possible value for the densities  $ s(t) \,  , i(t)$ and  $ r(t)$ is  $\frac{1}{N}$ and if any of the densities becomes smaller than $\frac{1}{N}$ the correspondent population  becomes extinct. This simple fact is not usually taken into consideration and it is of utmost importance for what follows.  
\begin{Remark} Introducing a new independent parameter defined as	$ \tau (t) =  \mu \, t $ 	the  SIR system becomes
	$$
	\begin{array}{l}
	\frac{ds}{d\tau} = - \frac{\beta}{\mu} \, s \, y +  \, 1 -  \, s   \, ,\\
	\\
	\frac{dy}{d\tau}  = \frac{\beta}{\mu} \, s \, y -  \, y \left( 1 - \frac{\gamma}{\mu} \right)  \, .\\
	\end{array}
	$$
	Therefore, the qualitative dynamics can be studied by fixing a $\mu$ value and considering the $\gamma$ and $\beta$ parameters as measured in $\mu$-units.
\end{Remark}

\subsection{Fade out for the SIR model}
The disease is said to fade out if it disappears in a finite time after some epidemic outbreak. For the deterministic SIR model this is only possible if for some time instant the number of  infected individuals is smaller then one, or equivalently if $y(t) < 1 / N $.

\begin{Definition}
Let $R_0 > 1$. The disease will {\it fade out} if  for some time instant  $t^*$
$$y(t^*) = 1/N \, \, \, \text{and} \, \, \, y^{\, \prime}(t^*) < 0 \, . \label{init2}$$ 
\end{Definition}
The $R_0 > 1 $ condition allows for  the possibility of epidemic outbreaks.\par 
 In what follows the fade out phenomena for the SIR will be characterized for the initial conditions
 \begin{equation} s(0)= 1 - \frac{1}{N} \, \, \, , \, \, \, y(0)=\frac{1}{N} \, \, \, \text{and} \, \, \, r(0)=0 \label{init} \end{equation}
 representing the invasion of a totally susceptible community by only one infected individual.
The characterization of the fade out for the SIR model is done by noticing that the local minima for the $y(t)$ curve form a monotonic increasing  sequence (see Appendix A for a proof of this fact). Therefore, in the case it exists, the first local minimum of $y(t)$ is also the global minimum for $ t > 0 $. 

\begin{Definition} Assume $R_0 = \frac{\beta}{\gamma + \mu} > 1 \, .$
Define	$\Psi = \Psi(\beta , \gamma , \mu , N)$ as the value of the first minimum of
	$y(t)$ with $t>0$,  for the initial conditions \ref{init}. If
	$y(t)$ has no minimum value for $t> 0 $ set $\Psi = y^* \equiv \frac{\mu}{\beta + \mu}\, ( 1 - 1 /R_0) $.
\end{Definition}

\begin{Definition} Let $\gamma$, $\beta$, $\mu $ all fixed and positives. $\tilde N > 0 $ is called the {\it Deterministic Critical Community Size} if 
	\begin{equation} 
	 \label{ccs}
	 \Psi(\beta , \gamma , \mu , \tilde N) = 1 /  \tilde N \, . 
	 \end{equation} 
\end{Definition}
If the population is smaller then $\tilde N$ the disease will  fade out and if the population is greater or equal to $\tilde N$ the number of infected individuals will converge to its limit value  $y^*$. The numerical procedure for the determination of the CCS is straightforward: Let $\alpha \, , \, \beta $ and $\gamma$ such that $R_0 > 1$.  Fixing a value for $N$, integrate  system (\ref{eq.sir_normalizado}) up to the first $y(t)$ local minimum (with $t>0$). If the minimum value is equal to $1 / N$ then $N$ is the CCS. If not, change the $N$ value and proceed on the same way. A bisection method can be used to determine the CCS.\par
\begin{figure}[H]    
	\includegraphics[width=10.8 cm, height= 8 cm]{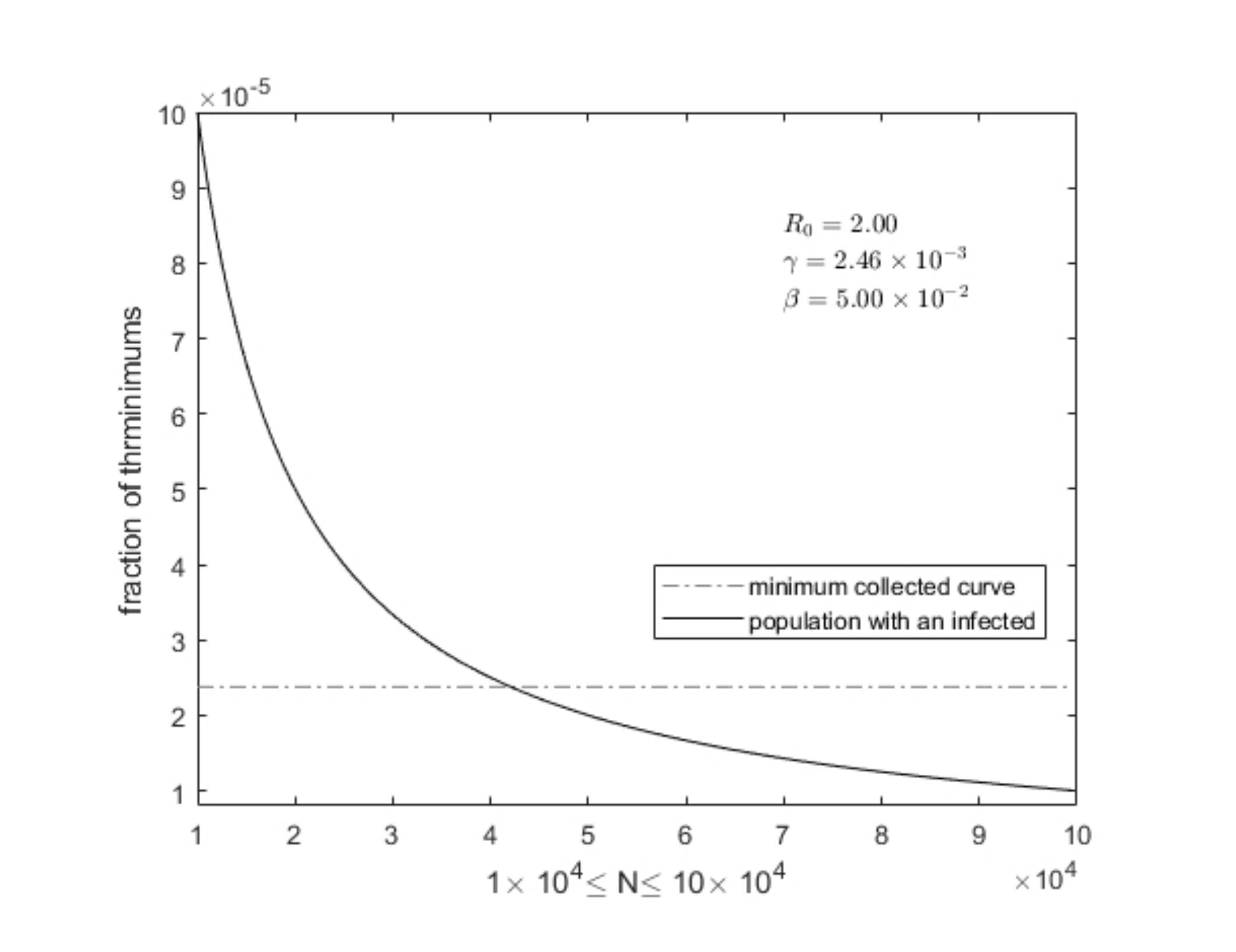}
	\caption{Numerical determination of the Critical Community Size (CCS).}
\end{figure} 

The CCS dependence on the parameters $\gamma$ and $\beta$ is showed on figure 
\begin{figure}[H] 
	\label{figccs}     
	\includegraphics[width=10.8 cm, height= 8 cm]{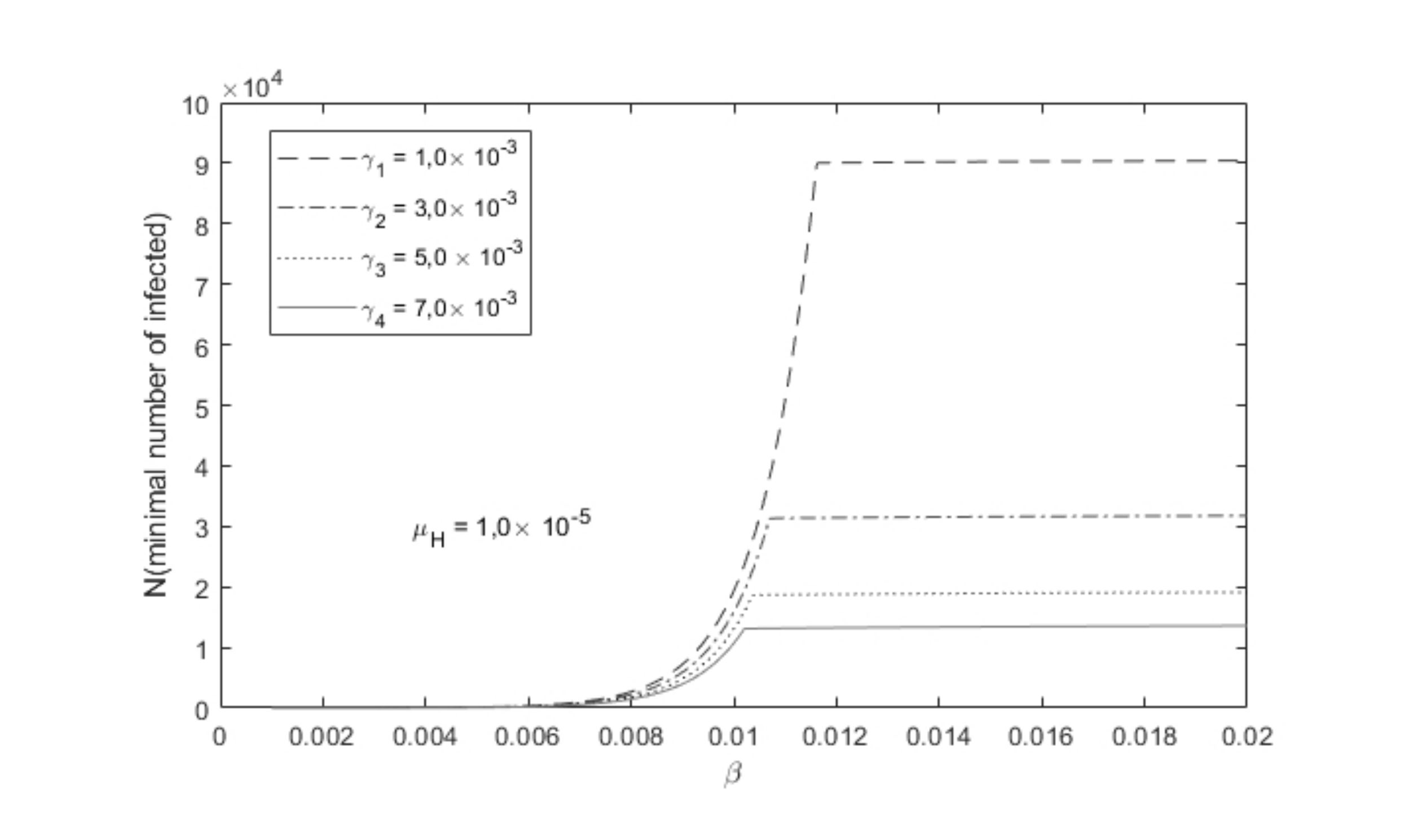}
	\caption{CCS  as a $\beta$ function.  $\gamma$ values are indicated on figure.}
\end{figure}

\begin{figure}[H]
	\label{figccs}   
	\centering
	\includegraphics[width=10.8 cm, height= 8 cm]{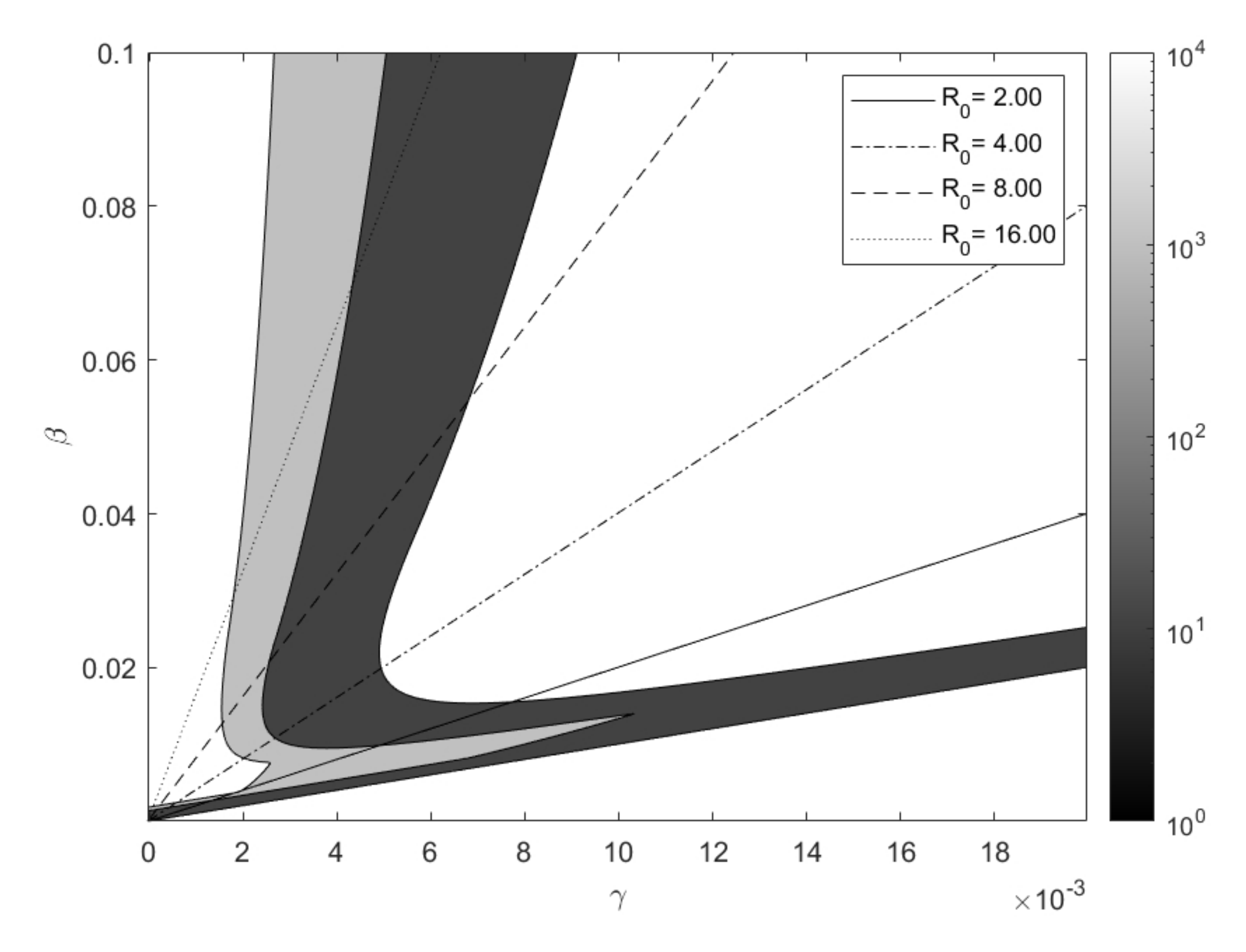}
	\caption{Critical Community Size (CCS) as a function of $\gamma$ and $\beta$. The level curves represent the CCS values. The mortality rate is $\mu = 1.0 \times 10^{-5}$.}
\end{figure} 

The numerical results are exhibited in figure (1) for $\mu = 1.0 \times 10^{-5}$. As an example notice that considering a population of 10000 persons all the parameters values in the yellow area will drive the disease to extinction in figure (3). The disease will persist for all parameters values outside the yellow region.
Also notice that  the level curves coalesce very near the $R_0 = 1 $ curve for larger values of the clearance rate. The consequence of this coalescence is that the CCS values for diseases with $R_0 \approx 1 $  can assume any value. The numerical results shows that the $R_0$ value, while determining the global dynamics of the SIR system, does not determine the CCS.

\section{Viral Competition under Environmental Selective Pressure}
\label{competition}
In this section  the implications of the CCS for a system with two competing virus strains will be explored. The main fact here is that the Principle of Competitive Exclusion (PCE) will no longer holds. \par 
Let $S(t)$, $I_i(t) \, \, (i=1,2) \, $ , $R(t)$ denote the numbers of individuals of susceptibles, infected by strain $i=1,2 \, $, and removed at time $t$ respectively. Let $N(t) = S(t) + I_1(t) + I_2(t) + R(t)$ be the total number of individuals on the population at time $t$. Let 
 $$s = S/N  \, \, \, , \, \, \, y_i= I_i/ N \,  i=1,2 \, \, \, , \, \, \, \text{e} \, \, \, r = R/N \, ,$$ be the respective densities.   $N(t) = N \in \mathbb{Z}$ since it represents the total number of individuals on the population. Accordingly, the smallest possible value for the densities  $ s(t) \,  , i_i(t)$ and  $ r(t)$ is  $\frac{1}{N}$. In order to make the discussion more general, an Allee effect will be introduced: If $I_i(t) \le \rho $ for some critical community size $ \rho \ge \frac{1}{N} $ the disease will be considered extinct. $\rho$ is expected to depend on the population density, mixing, etc.
 
 The two-strain model is
$$
\begin{array}{l}
s^{\, \prime} = - \beta_1 \, ( y_1 + y_2 ) \, s   + \mu \,( 1 - s )   \, ,\\
\\
y_1^{\, \prime} = \beta_1 \, y_1 \, s   - \gamma_1 \, y_1 - \mu \, y_1 \, ,\\
\\
y_2^{\, \prime} = \beta_2 \, y_2 \, s   - \gamma_2 \, s_2 - \mu \, s_2  \, ,\\
\\
r^{\, \prime} = \gamma_1 \, s_1  + \gamma_2 \, s_2 - \mu \, r \, s
\end{array}
$$
 The CCS for the above system is difficult to determine since the minima are no longer monotonically ordered and the parameter space is larger. To show that the basic reproductive number of the virus is no longer the unique factor deciding the outcome of the viral competition the first minima of $I_1(t)$ and $I_2(t)$ are calculated for the following situation: The parameters of the virus-1 are held constant at $R_{1} = 3,20 \, , \, \beta_1= 9,0 \times 10^{-3} \, , \, \mu_1 = 5 \times 10^{-5} \, . $ For the virus-2 $R_2 = 3,20 \, , \, \mu_2 = 5 \times 10^{-5} \, $ and $\beta_2$ varying from $4,0 \times 10^{-3}$ to $19,0 \times 10^{-3}$. The noticeable point here is the inversion of the infected minima curves. For $\beta_2 < 3.0 \times 10^{-3} $ the minima of $y_1$ are greater then $y_2$ and for $\beta_2 > 3.0 \times 10^{-3}$ the opposite holds. For a population threshold of $1 \times 10^{-5}$ the virus-2 is eliminated if $\beta_2 < 3.0 \times 10^{-3} $ and 
 the virus-1  can be eliminated if $\beta_2 >  3.0 \times 10^{-3} $. Notice the value o $R_0$ has been held constant for both viruses. The simulation is not contradicting the PCE since both viruses have the same $R_0$ value.\par 
	\begin{figure}[H] 
		\label{figccs}     
	\includegraphics[width=10.8 cm, height= 8 cm]{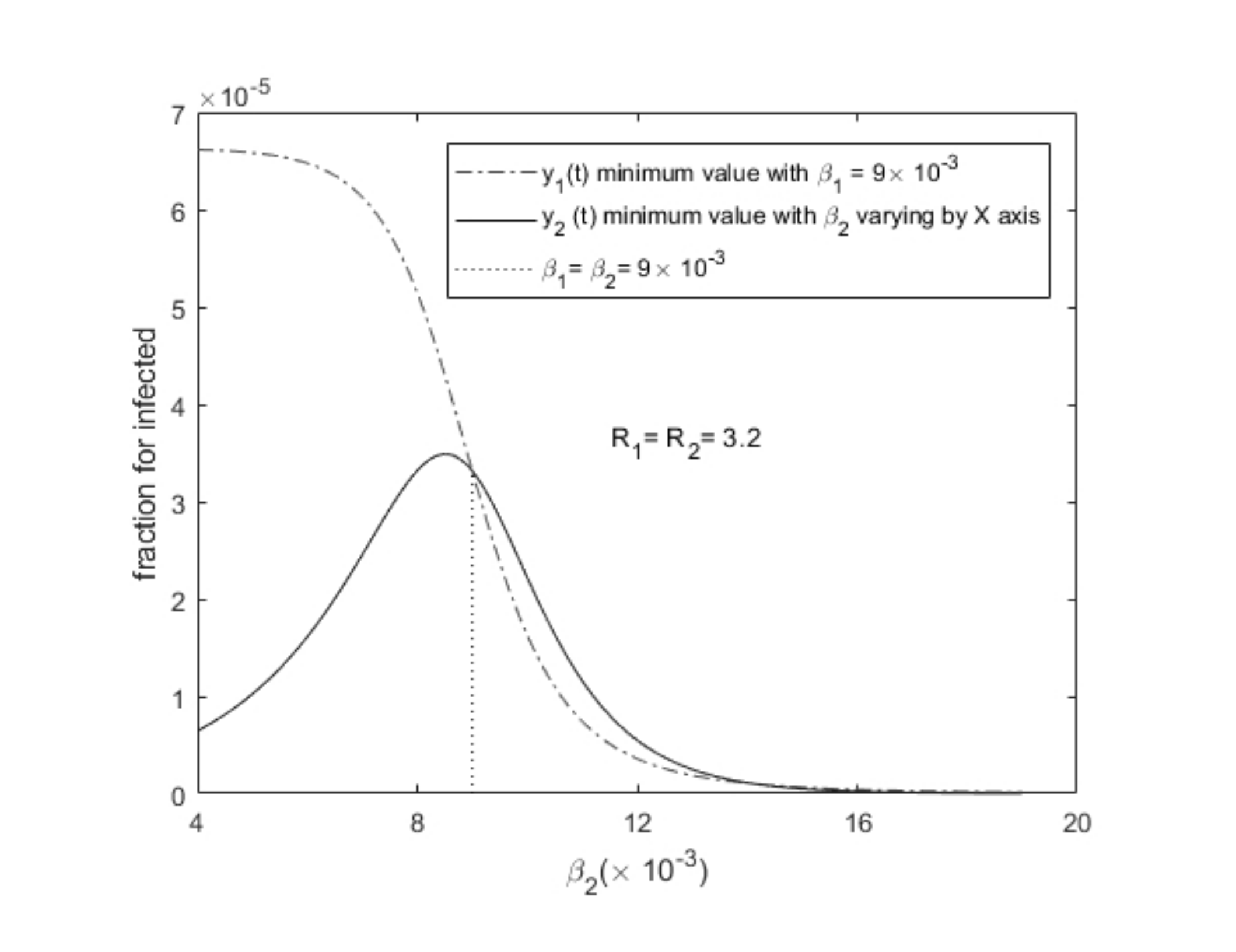} 
	\caption{First minima for the two virus strains as a $\beta_2$ function and fixed $\mu_H= 5\times 10^{-5}$.}
	\end{figure} 
	
		\begin{figure}[H] 
		\label{figccs2}     
		\includegraphics[width=10.9 cm, height= 8 cm]{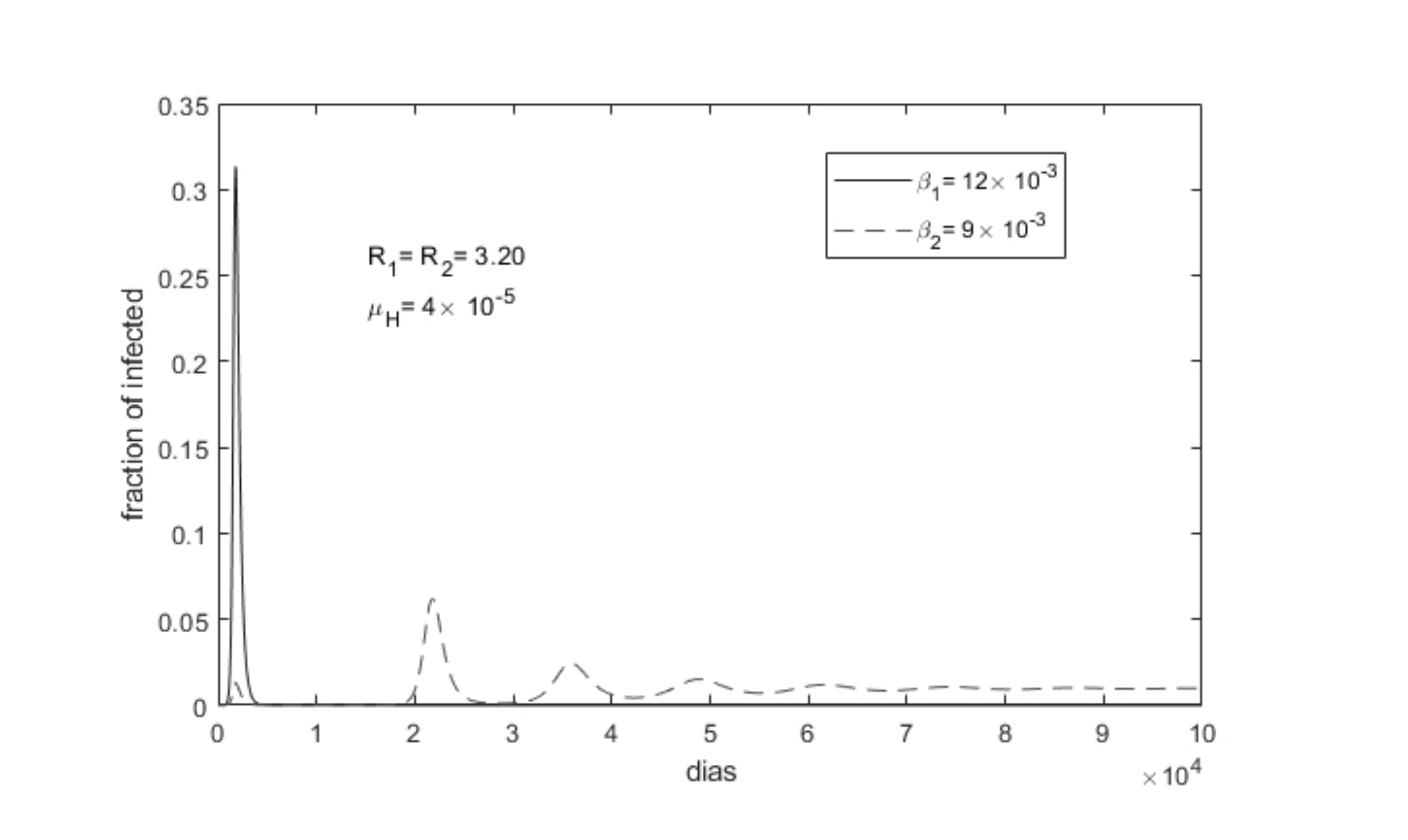}
	\caption{Elimination of the virus-1 after passing trough its first minimum.}
	\end{figure} 
	
	\begin{Remark}
	The example in figure (5) began with two people with a virus-1 and virus-2 strains each. Although after some time the strain with $\beta_1= 12\times 10^{-3}$ have no one infected individual. Therefore we can conclude that one strain go to extinction.     
	\end{Remark}
	
	The next simulation shows that the value of the first minimum 
	can be smaller for the virus with greater $R_0$ leading to its eradication and the permanence of 
	virus with smaller $R_0$ in contradiction to the PCE. 
	
	\begin{figure}[!h] 
	\label{figccs}     
	\includegraphics[width=10.8 cm, height= 9 cm]{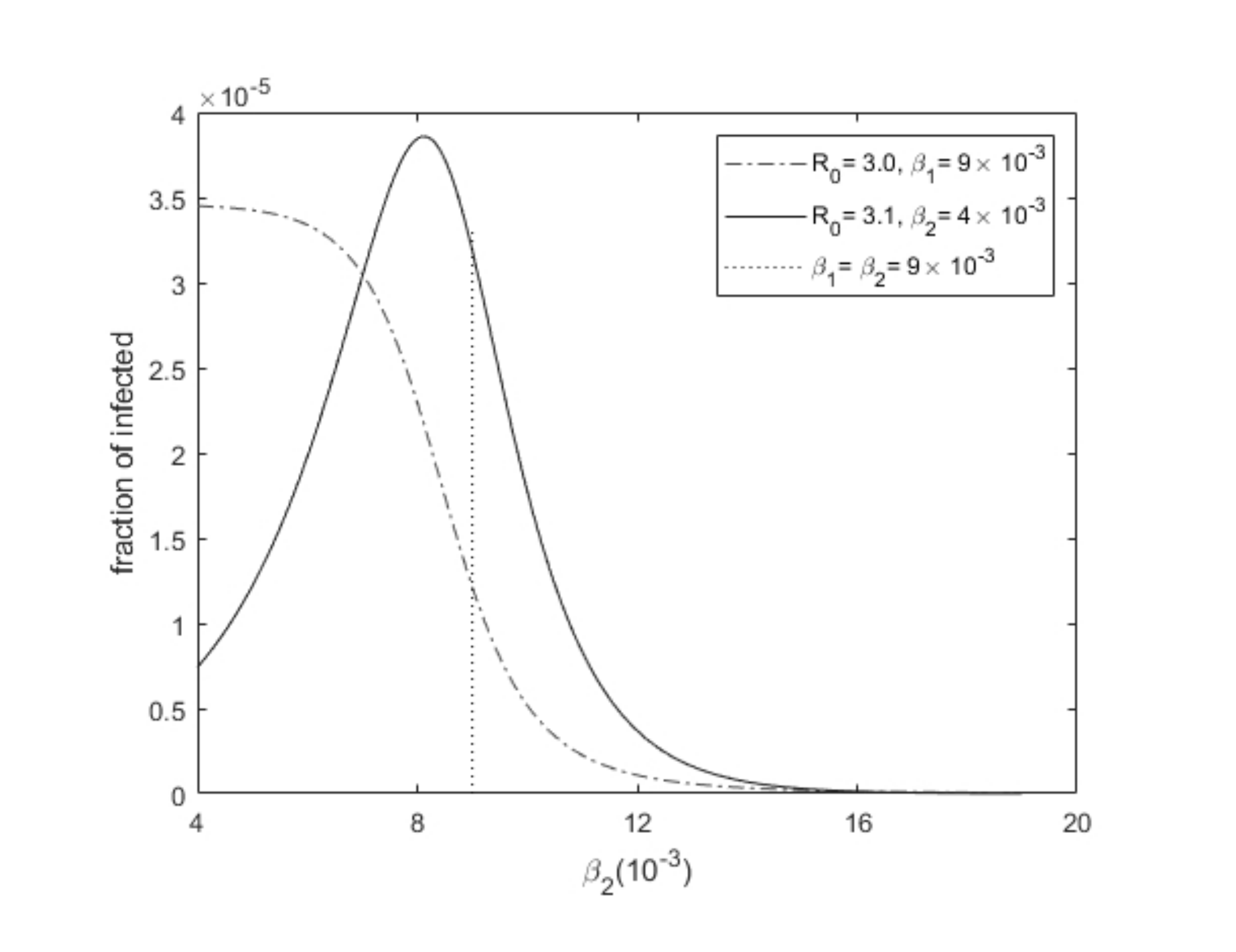} 
	\caption{First minima for the two virus strains as a $\beta_2$ function.}
\end{figure}

\begin{figure}[!h] 
	\label{figccs}     
	\includegraphics[width=10.8 cm, height= 9 cm]{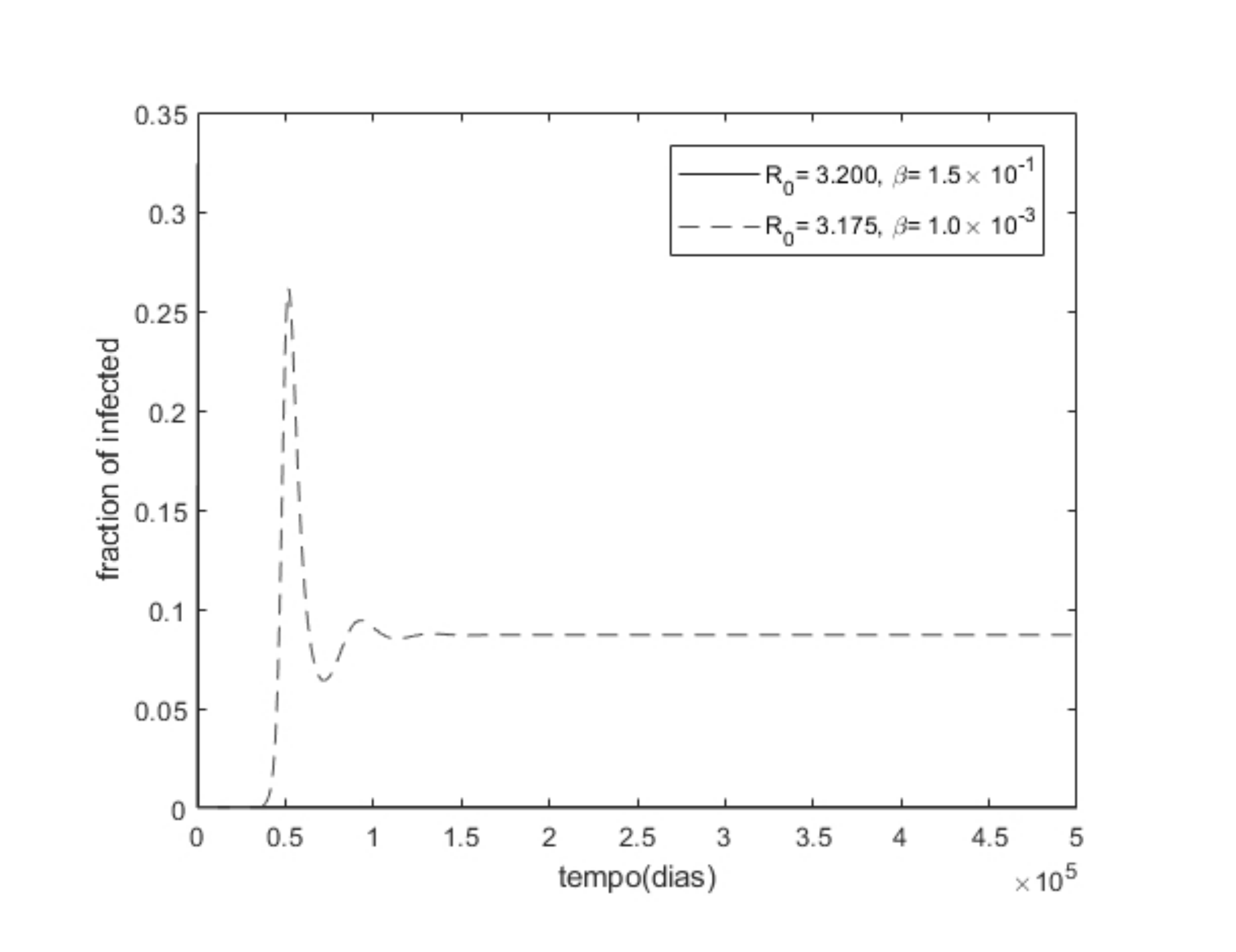}
	\caption{Elimination of the virus-1 after passing trough its first minimum.}
\end{figure}

\section{Conclusions}
The concept of a deterministic Critical Size Community, as introduced in this paper, provides 
an alternative potential mechanism for the disease fade out phenomena. It also allows the characterization of the disease persistence within a community in terms of  population size and disease parameters. 
This characterization is shown in Figure (\label{figccs}). The graphic shows for example, that two diseases with the same basic reproductive number $R_0$ can have rather different dynamics. This is, from a mathematical point of view, consequence of the non-linear scaling properties of the SIR system and has drastic consequences for the disease persistence.\par 

 The results also provide evidence that the Critical Community Size is an important component for the viral competition and natural selection processes. While for the PCE, the virus with higher $R_0$ will eliminate  the other viruses with smaller $R_0$, the concept of CCS shows that the dynamic is in fact more complex. The CCS allows for the extinction of any virus population reaching the minimum threshold value. While still predicting that only one virus type will persist in the long term, it shows that the survivor type needs no longer to be the one with the highest $R_0$. \par 
 
\label{conclusions}

\section{APPENDIX A}
\begin{Proposition}
	The time ordered local minima of $y(t)$ form a monotonic increasing sequence. Analogously the time ordered local maxima form a monotonic decreasing function.
\end{Proposition}
\begin{proof} Since 
 $\displaystyle \frac{\p^2 V(s,y)}{\p y^2} = \frac{y^*}{y^2} > 0 $ 
 the $E_1$ Lyapounov function is convex in $y$ \cite{korobeinikov2002lyapunov}.  \par 
 Let $y(t_1)$ and $y(t_2)$ be two consecutive local minima of $y(t)$. Since they are minima $y^{\, \prime}(t_i)=0$ and $y^{\, \prime \prime}(t_i) > 0 $ , $i=1,2$. This implies $s(t_i)=\frac{1}{R_0}$ and $y(t_i) < y^*$, $i=1,2$ . Since $V(s,t)$ is a Lyapounov function, it  is a decreasing function along the flow. Therefore   $t_2 > t_1$ implies that $V(1/R_0,y(t_2)) < V(1/R_0,y(t_1))$ 
 Also, $E_1$ is a global minimum of $V(s,y)$, it follows by convexity that
 $V(1/R_0,y)$ is a $y$ decreasing function for $y < y^*$ and therefore
 $y(t_2) > y(t_1) \, .$ The proof for the maxima is similar.

\end{proof}

\bibliographystyle{amsplain}
\nocite{*}
\bibliography{artigo}

\end{document}